\documentclass[10pt, doublecolumn]{IEEEtran}
\usepackage{epsfig,latexsym}
\usepackage{float}
\usepackage{indentfirst}
\usepackage{amsmath}
\usepackage{bm}
\usepackage{amssymb}
\usepackage{times}

\usepackage{algorithm}
\usepackage[noend]{algpseudocode}
\usepackage{subfigure}
\usepackage{psfrag}
\usepackage{hyperref}
\usepackage{cite}
\usepackage{lastpage}
\usepackage{fancyhdr}
\usepackage{color}
 \usepackage{amsthm}
\usepackage{bigints}
\usepackage{booktabs}
\usepackage{multirow}
\usepackage{siunitx}

\newtheorem{Lemma}{Lemma}
\newtheorem{Corollary}{Corollary}
\newtheorem{lemma}[Lemma]{$\mathbf{Lemma}$}

\newtheorem{corollary}[Corollary]{$\mathbf{Corollary}$}

\newcounter{problem}
\newcounter{save@equation}
\newcounter{save@problem}
\makeatletter

\begin{document}
\title{ \vspace{-0.8em}\LARGE{Resolution of Near-Field Beamforming and Its Impact on NOMA }}

\author{ Zhiguo Ding, \IEEEmembership{Fellow, IEEE}   \thanks{ 
  
\vspace{-1.5em}

 Z. Ding is with Department of Electrical Engineering and Computer Science, Khalifa University, Abu Dhabi, and   also with Department of Electrical and Electronic Engineering, University of Manchester, Manchester, UK.    

  }\vspace{-2.6em}}
 \maketitle

\vspace{-1em}
\begin{abstract}
The resolution of near-field beamforming is an important metric to measure how effectively   users with different locations can be distinguished. 
This letter identifies the condition under which the resolution of near-field beamforming is not perfect.  This imperfect resolution means that a legacy user's preconfigured near-field beam can be still useful to other users, which motivates the application of  non-orthogonal multiple access (NOMA). Both the   analytical and simulation results are presented  to demonstrate that indeed those near-field beams preconfigured for legacy users can   be used  to  serve additional NOMA users effectively, which improves the overall system throughput and connectivity. 
\end{abstract}\vspace{-0.5em}

\vspace{-1em} 

 \section{Introduction}
 Near-field beamforming  has recently received a lot of attention, due to the use of high carrier frequency bands, such as millimeter-wave (mmWave) and Terahertz  (THz) bands, which make   the Rayleigh distance   significantly large \cite{7400949,devar}. One exciting feature of near-field beamforming is that users with different locations can be perfectly  distinguished,  as explained by the following downlink multi-input single-output (MISO) example. For conventional far-field beamforming, the simplified beam-steering vector is used to model a user's  channel vector, and hence   two users with identical angles of departure but different distances from the base station share the same beam-steering vector, which makes it difficult for the base station  to distinguish the users \cite{7561012}. On the contrary, near-field beamforming can ensure that the two users are effectively  distinguished due to  the use of    the sophisticated  spherical-wave propagation model \cite{BJORNSON20193, 9738442}. This exciting  feature has led to the recent studies for   resource allocation in near-field communication networks, the improvement of degrees of freedom, and also the design of location based multiple access (LDMA) \cite{9738442,ldmadai, yuanweinear}. 

The aim of this letter is to provide a detailed study for the feasibility of near-field beamforming to distinguish  randomly located users. To facilitate the performance analysis, we first define the resolution of near-field beamforming in this letter, which is the cross-correlation of the users'  spherical-wave channel vectors.  The existing literature shows that the resolution of near-field beamforming becomes perfect, if the number of the antennas at the base station becomes infinity \cite{ldmadai}. In this letter, it shown that, unless the users are clustered very close to the base station, the resolution of near-field beamforming is not perfect.  In particular, analytical results are developed to show that for the users which have identical angles of departure,  the correlation of their channel vectors becomes almost one, i.e., the resolution of near-field beamforming is poor, if their distances to the base station are proportional to the Rayleigh distance. This conclusion holds even if  the number of antennas at the base station goes to infinity. The imperfect resolution of near-field beamforming means that  one user's  beam can be still useful to the others, which motivates the second contribution of this letter   for investigating  the feasibility of applying non-orthogonal multiple access (NOMA) in near-field networks \cite{10129111}. In particular, consider that there exists a near-field legacy network, where spatial  beams have been preconfigured for the legacy users. The carried-out feasibility study is to serve additional NOMA users by using these preconfigured beams.  Both the   analytical and simulation results are presented  to show that these preconfigured near-field beams   can indeed be used  to effectively serve additional NOMA users, which improves the overall   system throughput and connectivity. 

\section{System Model}
Consider a downlink network  with  $M$ legacy users, denoted by ${\rm U}_m^{\rm L}$, $m\in\{0, \cdots, M\}$, where each user is equipped with a single antenna, and the base station uses a uniform linear array (ULA) with $N$ elements. ${\rm U}_m^{\rm L}$'s observation is given  by $
y_m^{\rm L} = \mathbf{h}_m^H\sum^{M}_{i=1}\mathbf{p}_ix_i+n_m^{\rm L}$,
where  ${\rm U}_m^{\rm L}$'s channel vector is denoted by  $ \mathbf{h}_m$ and    based on the spherical-wave propagation model  because the legacy users' distances to the base station are assumed to be smaller than the Rayleigh distance, denoted by $d_{\rm Ray} $\cite{devar,9738442,Eldar2}, i.e., $\mathbf{h}_m= \sqrt{N \gamma_{{\boldsymbol \psi}^{\rm L}_m}}\mathbf{b}\left({\boldsymbol \psi}^{\rm L}_m\right)$, ${\boldsymbol \psi}^{\rm L}_m$ denotes the $m$-th legacy user's location based on the Cartesian coordinate, 
 \begin{align}\label{near}
\mathbf{b}\left({\boldsymbol \psi} \right) = \frac{1}{\sqrt{N}} \begin{bmatrix} 
 e^{-j\frac{2\pi }{\lambda}\left| {\boldsymbol \psi}  -{\boldsymbol \psi}_1\right|} &\cdots &  e^{-j\frac{2\pi }{\lambda}\left| {\boldsymbol \psi}  -{\boldsymbol \psi}_N\right|}
 \end{bmatrix}^T,
 \end{align} 
  ${\boldsymbol \psi}_n$ denotes the location of the $n$-th element of ULA,   $\gamma_{{\boldsymbol \psi}^{\rm L}_m} = \frac{c^2}{16\pi^2 f_c^2 \left| {\boldsymbol \psi}^{\rm L}_m -{\boldsymbol \psi}_0\right|^2}$, $c$, and $f_c$ denote the free-space path losses,   the speed of light, and    the carrier frequency, respectively, $n_m^{\rm L}$ denotes the additive Gaussian noise with power $P_N$, $\mathbf{p}_m$ denotes the beamforming vector, and $x_m$ denotes the signal sent on $\mathbf{p}_m$.  The polar coordinates of the  legacy   and   NOMA users' locations  are denoted by $(r^{\rm L}_m,\theta^{\rm L}_m)$,  and  $(r^{\rm N}_k,\theta^{\rm N}_k)$, respectively.  We note that, in \eqref{near},     the path losses from different elements of the ULA  to a user is assumed to be almost  same, since the users' distances to the base station are considered to be proportional to the Rayleigh distance in this letter.  

As in \cite{10129111}, it is assumed that $\mathbf{p}_m$ has been  preconfigured according to the $m$-th legacy user's channel vector, and the aim of the letter is to investigate the feasibility to serve additional NOMA users by using the preconfigured beams. As commonly used in the massive MIMO literature,    the maximal ratio combining type of precoding is used, i.e., $\mathbf{p}_m=\mathbf{b}\left({\boldsymbol \psi}^{\rm L}_m\right)$, which means that $\mathbf{p}_m^H\mathbf{p}_m=1$, and $|\mathbf{p}_m^H\mathbf{h}_m|^2= N \gamma_{{\boldsymbol \psi}^{\rm L}_m} $. 

Consider that there are $K$ NOMA users, denoted by ${\rm U}_k^{\rm N}$, $k\in\{1, \cdots, K\}$, and ${\rm U}_k^{\rm N}$ receives the following observation: $
y_k^{\rm N} = \mathbf{g}_k^H\sum^{K}_{i=1}\mathbf{p}_ix_i+n_k^{\rm N}$,
where $\mathbf{g}_k$ and $n_k^{\rm N}$ are defined similar to $\mathbf{h}_m$ and $n_k^{\rm L}$. 

To reduce the system complexity, assume that at most a single NOMA user is scheduled on each beam. If    ${\rm U}_k^{\rm N}$ is scheduled on   beam $\mathbf{p}_m$,    the base station sends $x_m =\sqrt{P_{S}\alpha_m^{\rm L}}s_m^{\rm L}+\sqrt{P_{S}\alpha_m^{\rm N}}s_k^{\rm N}$ on this beam, where $P_{S}$ denotes the transmit power budget on each beam,  $s_m^{\rm L}$ and $s_k^{\rm N}$ denote the signals to   ${\rm U}_m^{\rm L}$ and  ${\rm U}_k^{\rm N}$, respectively,   $\alpha_m^{\rm L}$ and $\alpha_m^{\rm N}$ denote the NOMA power allocation coefficients, and $ \alpha_m^{\rm N}+ \alpha_m^{\rm L}=1$. To avoid disruption to the legacy users' quality of experience, the NOMA users' signals will be decoded first on each beam,  and  the data rate of the NOMA user,  ${\rm U}_k^{\rm N}$, on beam $\mathbf{p}_m= \mathbf{b}\left({\boldsymbol \psi}^{\rm L}_m\right)$ is capped  by {$
R^{\rm N}_k = \log\left( 
1+\frac{
P_{S}N\gamma_{{\boldsymbol \psi}^{\rm N}_k}\alpha_m^{\rm N}|\mathbf{b}\left({\boldsymbol \psi}^{\rm L}_m\right)^H\mathbf{b}\left({\boldsymbol \psi}^{\rm N}_k\right)|^2
}{P_{S}N\gamma_{{\boldsymbol \psi}^{\rm N}_k}\alpha_m^{\rm L}|\mathbf{b}\left({\boldsymbol \psi}^{\rm L}_m\right)^H\mathbf{b}\left({\boldsymbol \psi}^{\rm N}_k\right)|^2+I^{\rm N}_k+P_N}
\right)$}, where $I^{\rm N}_k=\sum^{M}_{\substack{i=1\\ i\neq m}}P_SN\gamma_{{\boldsymbol \psi}^{\rm N}_k}|\mathbf{b}\left({\boldsymbol \psi}^{\rm L}_i\right)^H\mathbf{b}\left({\boldsymbol \psi}^{\rm N}_k\right)|^2$.
The legacy user is expected to have a strong channel gain, and therefore is capable to carry   out successive interference cancellation (SIC)   by decoding  its partner's signal with the following data rate: $
\tilde{R}^{\rm L}_m = \log\left( 
1+\frac{
P_{S}N\gamma_{{\boldsymbol \psi}^{\rm L}_m} \alpha_m^{\rm N}|\mathbf{b}\left({\boldsymbol \psi}^{\rm L}_m\right)^H\mathbf{b}\left({\boldsymbol \psi}^{\rm L}_m\right)|^2
}{P_{S}N\gamma_{{\boldsymbol \psi}^{\rm L}_m} \alpha_m^{\rm L}|\mathbf{b}\left({\boldsymbol \psi}^{\rm L}_m\right)^H\mathbf{b}\left({\boldsymbol \psi}^{\rm L}_m\right)|^2+I_m^{\rm L}+P_N}
\right)$, where $I_m^{\rm L}=\sum^{M}_{\substack{i=1\\ i\neq m}}P_SN\gamma_{{\boldsymbol \psi}^{\rm L}_m} |\mathbf{b}\left({\boldsymbol \psi}^{\rm L}_i\right)^H\mathbf{b}\left({\boldsymbol \psi}^{\rm L}_m\right)|^2$. If the first stage of SIC is successful, the legacy user can decodes its own message with the following data rate:
$
R^{\rm L}_m = \log\left( 
1+\frac{
P_{S}N\gamma_{{\boldsymbol \psi}^{\rm L}_m} \alpha_m^{\rm L}|\mathbf{b}\left({\boldsymbol \psi}^{\rm L}_m\right)^H\mathbf{b}\left({\boldsymbol \psi}^{\rm L}_m\right)|^2
}{ I_m^{\rm L}+P_N}
\right)$.

\section{Resolution of Near-Field Beamforming } 

The resolution of near-field beamforming can be defined as   $\Delta \triangleq  \left|\mathbf{b}\left({\boldsymbol \psi}_1\right)^H\mathbf{b}\left({\boldsymbol \psi}_2\right)\right|^2$,  for any ${\boldsymbol \psi}_1\neq {\boldsymbol \psi}_2$.   As can be seen from the expressions of $R^{\rm N}_k$, $\tilde{R}^{\rm L}_m$, and $R^{\rm L}_m$,  $\Delta$ has an important impact on these achievable data rates. 


Define $\left(r_i,\theta_i \right)$ as the polar coordinates corresponding to the Cartesian coordinates ${\boldsymbol \psi}_i$, $i\in\{1,2\}$.  We note that the resolution in the angle  domain, i.e., $\Delta$ with $r_1=r_2$, has been well studied in the mmWave literature \cite{Zhiguo_mmwave}, and therefore this letter focuses on the resolution in the distance  domain, i.e., $\Delta$ with $\theta_1=\theta_2$. In \cite{ldmadai},  the resolution of near-field beamforming in the distance domain is shown to be asymptotically perfect, i.e., $\Delta \triangleq  \left|\mathbf{b}\left({\boldsymbol \psi}_1\right)^H\mathbf{b}\left({\boldsymbol \psi}_2\right)\right|^2\rightarrow 0$, for $N\rightarrow \infty$, if $\theta_1=\theta_2$. The following lemma illustrates that the resolution of near-field beamforming in the distance domain is poor, if $r_i$, $i\in\{1,2\}$, are proportional to the Rayleigh distance.  
 
 \begin{lemma}\label{lemma1}
Assume that $\theta_1=\theta_2\triangleq\theta_0$ and $r_i=\beta_id_{\rm Ray}$. If $\tau\triangleq  (1-\sin^2\theta_0) \left(\frac{  1}{\beta_1}-\frac{ 1}{\beta_2}\right)\rightarrow 0$, $\Delta$ can be approximated as follows: 
\begin{align}
\Delta 
\approx&     
     1-\frac{ \pi^2\tau^2 (N+1)\left(
    26N^2-38    \right) }{5760(N-1)^3}.
\end{align}
  \end{lemma}
\begin{proof}
See Appendix \ref{proof1}.
\end{proof}
Define $f(x) =\frac{   (x+1)\left(
    26x^2-38    \right)}{ (x-1)^3}$ whose first-order derivative is given by
$f'(x) = \frac{-104x^2+24x+152}{(x-1)^4}$.  The largest root of $f'(x)=0$ is $1.33$, which means $f'(n)<0$, for $n=\{ 2,3,\cdots \}$. Therefore,   $f(N)$ is a monotonically decreasing function of $N$, which leads to the following corollary. 
\begin{corollary}\label{corollary1}
Assume that $\theta_1=\theta_2\triangleq\theta_0$ and $r_m=\beta_md_{\rm Ray}$. If $\tau\triangleq  (1-\sin^2\theta_0) \left(\frac{  1}{\beta_1}-\frac{ 1}{\beta_2}\right)\rightarrow 0$, $\Delta$ is a monotonically increasing function of $N$, and with $N\rightarrow \infty$, $\Delta$ becomes a constant as follows:
\begin{align}
\Delta 
 \rightarrow   1-\frac{13 \pi^2\tau^2  }{2880 } . 
\end{align}
\end{corollary} 
{\it Remark 1:}   Lemma \ref{lemma1} and Corollary \ref{corollary1} indicate that the resolution of NF beamforming is not perfect in the distance domain. This features facilitates the implementation of NOMA, where a beam preconfigured to a legacy user can be still useful to those users which are located on the direction of the beam. 

  {\it Remark 2}: We note that $\tau $ is small if the two users are relatively close and  $\theta$ is moderately  large.  For example, for the case with $\beta_1=0.5$, $\beta_2=0.7$ and $\theta=80^{\circ}$, $\tau\approx 0.017$, and for the case with $\beta_1=0.5$, $\beta_2=0.55$ and $\theta=50^{\circ}$, $\tau\approx 0.075$. Note that for the case with $\beta_1=0.5$ and $\beta_2=0.55$, the distance between the two users is still over $70$ m, if the carrier frequency is $28$ GHz and the base station has   $513$ antennas.

{\it Remark 3:} We note that the   analytical results in Lemma \ref{lemma1} and Corollary \ref{corollary1} are not general, but applicable to the special cases with $\tau\rightarrow 0$ only. In \cite{ldmadai}, it was stated that $\Delta 
 \rightarrow 0$ for $N\rightarrow \infty$, which is true only if one or both the users are very close to the base station, which makes Lemma \ref{lemma1} not applicable.  Nevertheless,  for   users which are not too  close from the base station, the resolution of near-field beamforming is limited, which motivates the implementation of NOMA in NF-SDMA networks, as shown in the following section.

 \section{A Feasibility Study for Application of NOMA}
 
Lemma \ref{lemma1} and Corollary \ref{corollary1} indicate that a preconfigured near-field beam can be still used to admit additional users which  are not at the focal point of the beam. Therefore, in this section, we focus on a special case with a single legacy user, i.e., $M=1$, where  the locations of the NOMA users follow a Poisson line process \cite{Haenggi}. In particular, assume that  the NOMA users' locations follow a one-dimensional homogeneous Poisson point process (HPPP) with density $\lambda$, i.e.,  all  the NOMA users are assumed to be  on the segment which is between the legacy user and the cell boundary and  aligned with the preconfigured beam, i.e., $ \theta^{\rm L}_1= \theta^{\rm N}_k\triangleq  \theta_0$, $k=1, 2, \cdots K$. A practical example is that the legacy user is on a street and the preconfigured beam is aligned with the street, where it is reasonable to assume that there are multiple users and devices on the street. 

  Furthermore,  assume that the NOMA users are ordered according to their distances to the legacy user, and ${\rm U}_k^{\rm N}$ is the $k$-th closest neighbour of the legacy user, i.e., $d_1\leq d_2\leq \cdots$, where   $d_k = \left| {\boldsymbol \psi}^{\rm N}_k -{\boldsymbol \psi}^{\rm L}_1\right|$.  The distance between ${\rm U}_1^{\rm L}$ and the base station is denoted by $r_{1}^{\rm L}$.   
The cumulative distribution function (CDF) of $d_k$ is the probability to have less than $k$ users on a segment with the length $d_k$ \cite{Haenggi}: 
\begin{align}
F_k(r) =&  1-\sum^{k-1}_{i=0}
e^{-\lambda r}\frac{\lambda^i r^i}{i!},
\end{align} 
which means that the probability density function (pdf) of $d_k$ can be obtained as follows: 
\begin{align}
f_k(r)  
=& e^{-\lambda r}\left(\sum^{k-1}_{i=0}
\frac{1 }{i!}
\lambda^{i+1}r^i 
-\sum^{k-1}_{i=0}
\frac{1 }{i!}i \lambda^ir^{i-1}
\right) 
 \\\nonumber
=& e^{-\lambda r} 
\frac{1 }{(k-1)!}
\lambda^{k}r^{k-1}. 
\end{align} 

For the considered special case with $M=1$,   the NOMA user's data rate can be simplified as follows: 
 $
R^{\rm N}_k = \log\left( 
1+\frac{
P_{S}N\gamma_{{\boldsymbol \psi}^{\rm N}_k}\alpha_1^{\rm N}|\mathbf{b}\left({\boldsymbol \psi}^{\rm L}_1\right)^H\mathbf{b}\left({\boldsymbol \psi}^{\rm N}_k\right)|^2
}{P_{S}N\gamma_{{\boldsymbol \psi}^{\rm N}_k}\alpha_1^{\rm L}|\mathbf{b}\left({\boldsymbol \psi}^{\rm L}_1\right)^H\mathbf{b}\left({\boldsymbol \psi}^{\rm N}_k\right)|^2 +P_N}
\right)$, which means that the outage probability for the NOMA user ${\rm U}^{\rm N}_k$ is given by
{\small \begin{align}
&\mathbb{P}^o =\sum^{k-1}_{i=0}\mathbb{P}(K=i) + \sum^{\infty}_{i=k}\mathbb{P}(K=i)\\\nonumber &\times  \mathbb{P}\left(\log\left( 
1+\frac{
P_{S}N\gamma_{{\boldsymbol \psi}^{\rm N}_k}\alpha_1^{\rm N}|\mathbf{b}\left({\boldsymbol \psi}^{\rm L}_1\right)^H\mathbf{b}\left({\boldsymbol \psi}^{\rm N}_k\right)|^2
}{P_{S}N\gamma_{{\boldsymbol \psi}^{\rm N}_k}\alpha_1^{\rm L}|\mathbf{b}\left({\boldsymbol \psi}^{\rm L}_1\right)^H\mathbf{b}\left({\boldsymbol \psi}^{\rm N}_k\right)|^2 +P_N}
\right)\leq R\right) ,
\end{align}}
\hspace{-0.5em}where $R$ denotes the NOMA users' target data rate, and $\mathbb{P}(K=i)$ denotes the probability of the event that there are $i$ NOMA users on the segment. 

 Assume that $\theta_0$ (or equivalently $\theta^{\rm L}_1$) is large and $\tau_k\triangleq  (1-\sin^2\theta_0) \left(\frac{  d_{\rm Ray}}{r_1^{\rm L}}-\frac{ d_{\rm Ray}}{r_1^{\rm L}+d_k}\right) $ is small. The use of Lemma \ref{lemma1} can simplify the expression of  the outage probability as follows:  
\begin{align}
\mathbb{P}^o  \approx&\sum^{k-1}_{i=0}\mathbb{P}(K=i) + \sum^{\infty}_{i=k}\mathbb{P}(K=i)\times\\\nonumber &
\mathbb{P}\left(\frac{
 \frac{\eta_2}{   \left(r^{\rm L}_1+r\right)^2}\alpha_1^{\rm N}\left(1-\eta_1  \left(\frac{  1}{r^{\rm L}_1}-\frac{ 1}{r^{\rm L}_1+r}\right)^2\right)
}{ \frac{\eta_2}{   \left(r^{\rm L}_1+r\right)^2}\alpha_1^{\rm L}\left(1-\eta_1  \left(\frac{  1}{r^{\rm L}_1}-\frac{ 1}{r^{\rm L}_1+r}\right)^2\right) +P_N}\leq \epsilon_1\right)
\\\nonumber =&\sum^{k-1}_{i=0}\mathbb{P}(K=i) + \sum^{\infty}_{i=k}\mathbb{P}(K=i)\\\nonumber &\times
\mathbb{P}\left(
 \frac{\eta_2}{   \left(r^{\rm L}_1+r\right)^2} \left(1-\eta_1  \left(\frac{  1}{r^{\rm L}_1}-\frac{ 1}{r^{\rm L}_1+r}\right)^2\right) \leq \epsilon_2\right),
\end{align}
where    $\epsilon_1=2^R-1$, $\epsilon_2 = \frac{P_N\epsilon_1}{\alpha_1^{\rm N}-\alpha_1^{\rm L}\epsilon_1}$,  $\eta_1=\frac{ \pi^2 (N+1)\left(
    26N^2-38    \right)}{5760(N-1)^3}(1-\sin^2\theta_0)^2d_{\rm Ray}^2$, and $\eta_2=\frac{P_{S}N c^2}{16\pi^2 f_c^2  }$. 

Denote the  four roots of the following equation by $z_1\geq\cdots\geq z_4$: 
\begin{align}
 -\eta_1x^4+\frac{  2\eta_1}{r^{\rm L}_1}x^3 +\left(1- \frac{  \eta_1}{\left(r^{\rm L}_1\right)^2}\right) x^2    - \frac{\epsilon_2}{\eta_2} =0.
\end{align}
We note that the roots are constants and    not related to the users' random locations. For the considered geometric scenario with large $d_{\rm Ray}$, we note that there are   two positive roots, i.e.,  $z_1\geq z_2>0$. Therefore, the outage probability can be calculated as follows:
\begin{align}
\mathbb{P}^o    
=&\sum^{k-1}_{i=0}\mathbb{P}(K=i) + \sum^{\infty}_{i=k}\mathbb{P}(K=i)\\\nonumber &\times
\left[\mathbb{P}\left(  R_D-r^{\rm L}_1\geq  r \geq \max\left\{0,\frac{1}{z_2} - r^{\rm L}_1\right\}\right)\right.
\\\nonumber &\left.+
\mathbb{P}\left(  r \leq \min\left\{ \max\left\{0,\frac{1}{z_1} - r^{\rm L}_1\right\} ,  R_D-r^{\rm L}_1\right\}\right)
\right]
\end{align}

By using the pdf of $d_k$, the outage probability is given by
\begin{align}
\mathbb{P}^o \approx&   
 \sum^{k-1}_{i=0}\mathbb{P}(K=i) + \sum^{\infty}_{i=k}\mathbb{P}(K=i)\frac{1 }{(k-1)!}  \\\nonumber &\times
\left(
 \gamma\left(k,\lambda \left(R_D-r^{\rm L}_1\right)\right)- \gamma\left(k,\lambda  \max\left\{0,\frac{1}{z_2} - r^{\rm L}_1\right\}\right) \right.
 \\\nonumber
 &+\left.\gamma\left(k,\lambda \min\left\{ \max\left\{0,\frac{1}{z_1} - r^{\rm L}_1\right\} ,  R_D-r^{\rm L}_1\right\}\right)\right),
\end{align}
where  \cite[equation 3.381]{GRADSHTEYN} is used and $\gamma(\cdot)$ denotes the incomplete gamma function.  

With some algebraic manipulations, the following corollary can be obtained for the outage probability.
\begin{corollary}\label{corollary2}
By assuming that the NOMA users follows the one-dimensional HPPP and $\tau_k \rightarrow 0$, the outage probability for the $k$-th nearest neighbour to the legacy user is given by
\begin{align}
\mathbb{P}^o \approx &\sum^{k-1}_{i=0}e^{-\lambda (R_D-r^{\rm L}_1)}\frac{\lambda^i (R_D-r^{\rm L}_1)^i}{i!}\\\nonumber &+ \sum^{\infty}_{i=k}e^{-\lambda (R_D-r^{\rm L}_1)}\frac{\lambda^i (R_D-r^{\rm L}_1)^i}{i!} \frac{1 }{(k-1)!}  
\\\nonumber &\times
\left(
 \gamma\left(k,\lambda \left(R_D-r^{\rm L}_1\right)\right)- \gamma\left(k,\lambda  \max\left\{0,\frac{1}{z_2} - r^{\rm L}_1\right\}\right) 
 \right.
 \\\nonumber
 &+\left.\gamma\left(k,\lambda \min\left\{ \max\left\{0,\frac{1}{z_1} - r^{\rm L}_1\right\} ,  R_D-r^{\rm L}_1\right\}\right)\right). 
\end{align} 
\end{corollary}  

 \section{Numerical Studies}
In this section, computer simulation results are presented to demonstrate the resolution of near-field beamforming, and its impact on the implementation of NOMA. For all the carried out simulations, the carrier frequency of $28$ GHz is used, the noise power is $-80$ dBm, and the antenna spacing for the ULA is set to be half of the wavelength.

In Figs. \ref{fig1} and \ref{fig2x}, the impact of $\theta_0$ on the resolution of near-field beamforming, $\Delta \triangleq  \left|\mathbf{b}\left({\boldsymbol \psi}_1\right)^H\mathbf{b}\left({\boldsymbol \psi}_2\right)\right|^2$,  for any ${\boldsymbol \psi}_1\neq {\boldsymbol \psi}_2$, is investigated. In Fig. \ref{fig1}, the two users' distances to the base station are proportional to the Rayleigh distance, i.e.,  $r_i=\beta_id_{\rm Ray}$, $i\in\{1,2\}$. In this case, Fig. \ref{fig1} shows that the resolution of near-field beamforming is poor, particular if $\theta_0$ is large. In addition, the figure also demonstrates that the new approximation results shown in Lemma \ref{lemma1} are more accurate that the existing one in \cite{ldmadai}, particularly for large $\theta_0$.    In Fig. \ref{fig2x}, the users are deployed   close to the base station, i.e.,  their distances to the base station are no longer proportional to the Rayleigh distance. In this case, Fig. \ref{fig2x} demonstrates that the resolution of near-field beamforming becomes accurate, particularly if the users are very close to the base station and the number of antennas becomes very large.  

  \begin{figure}[t]\centering \vspace{-0.5em}
    \epsfig{file=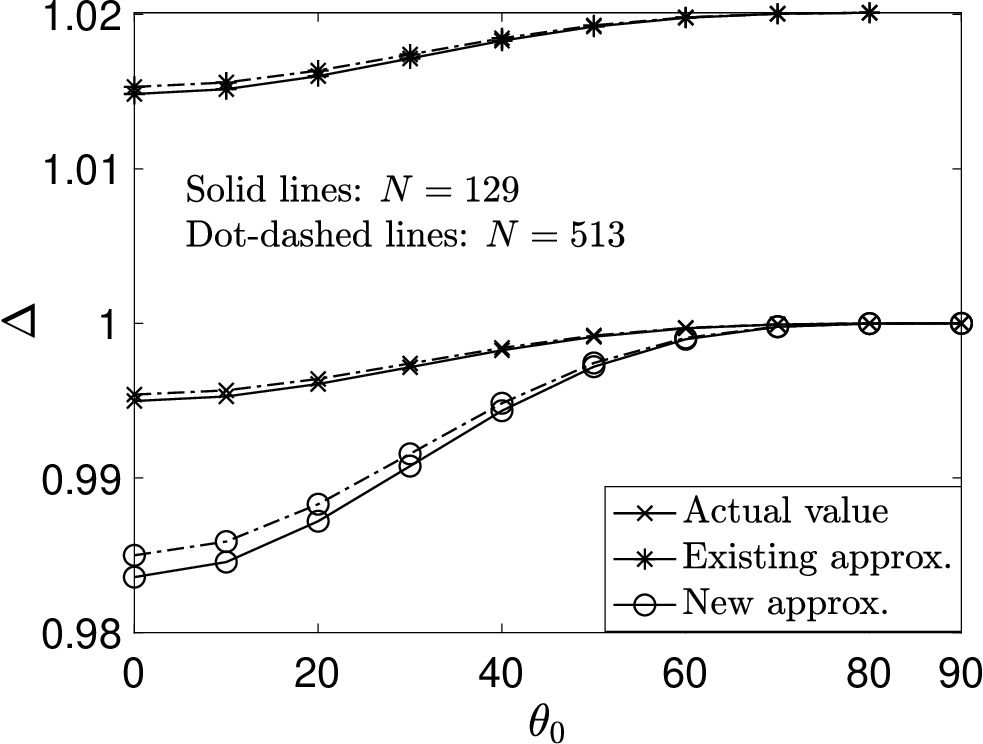, width=0.33\textwidth, clip=}\vspace{-0.5em}
\caption{Impact of $\theta_0$ on the resolution of near-field beamforming, $\Delta$, where $\theta_i=\theta_0$,  $r_i=\beta_id_{\rm Ray}$, $i\in\{1,2\}$, $\beta_1=0.5$ and $\beta_2=0.2+\beta_1$.  \vspace{-1em}    }\label{fig1}   \vspace{-0em} 
\end{figure}

   \begin{figure}[t]\centering \vspace{-0em}
    \epsfig{file=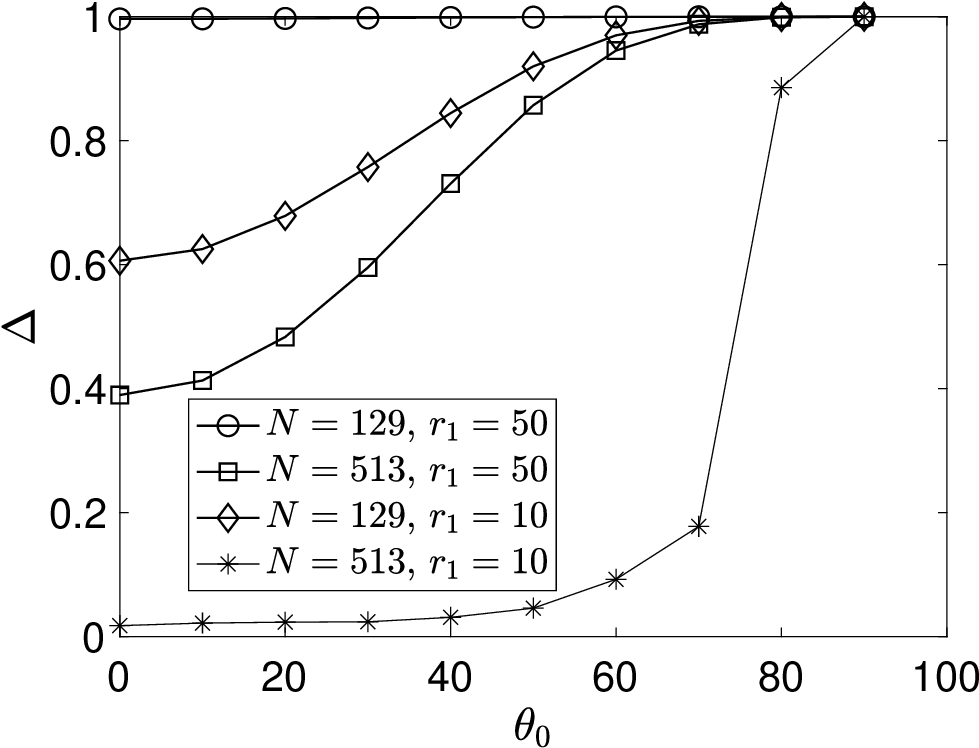, width=0.33\textwidth, clip=}\vspace{-0.5em}
\caption{Impact of $\theta_0$ on the resolution of near-field beamforming, $\Delta$, where $\theta_i=\theta_0$, $i\in\{1,2\}$, and $r_1-r_2=20$.  \vspace{-1em}    }\label{fig2x}   \vspace{-0.1em} 
\end{figure}

  \begin{figure}[t] \vspace{-0.1em}
\begin{center}
\subfigure[ $N=129$,  $\theta_0=45^0$]{\label{fig2a}\includegraphics[width=0.35\textwidth]{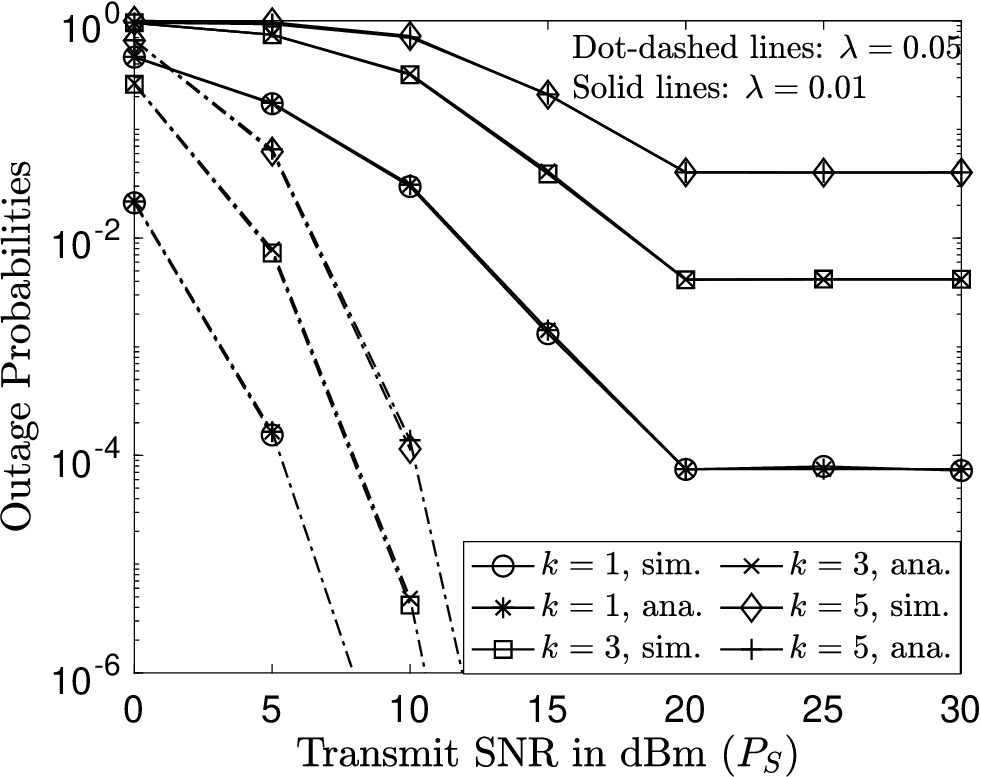}}\hspace{2em}
\subfigure[$\lambda=0.01$ and $k=1$]{\label{fig2b}\includegraphics[width=0.35\textwidth]{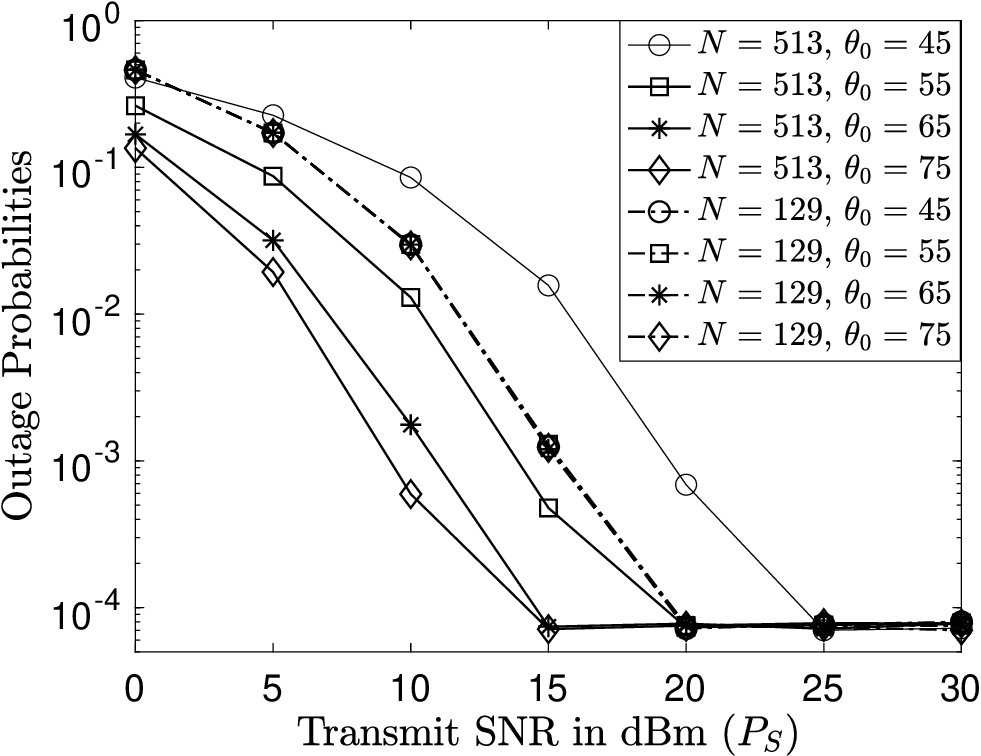}} \vspace{-1em}
\end{center}
\caption{Outage probability of the scheduled NOMA user by using the Poisson line process.  $M=1$, $r^{\rm L}_1=50$, $R_D=1000$, $R=0.5$ BPCU,    $\alpha_m^{\rm N}=\frac{4}{5}$ and  $\alpha_m^{\rm L}=\frac{1}{5}$. \vspace{-0.1em} }\label{fig2}\vspace{-2em}
\end{figure}

  Figs. \ref{fig1} and \ref{fig2x} show that the resolution of near-field beamforming is poor, unless the users are clustered close to the base station, which demonstrates the feasibility of using NOMA, as the users are most likely to be scattered   in the network. In Fig. \ref{fig2}, the NOMA users are assumed to be randomly deployed on the segment which is on the direction of ${\rm U}_1^{\rm L}$'s channel vector and between  ${\rm U}_1^{\rm L}$ and the boundary of the cell, where there is a single legacy user, $M=1$, the radius of the cell is denoted by $D_R$ and the $k$-th  closest neighbour to the legacy user  is scheduled. Fig. \ref{fig2a} shows that the scheduled  NOMA user's outage probability can be reduced significantly either by increasing the user intensity, $\lambda$, or by reducing the distance between the scheduled user and the legacy user, i.e., reducing $k$. Fig. \ref{fig2b}  shows that the user's outage probability can be further reduced by increasing $\theta_0$ for the case of $N=513$, whereas the impact of $\theta_0$ on the outage probability is insignificant for the case of $N=129$. Fig. \ref{fig3a} also demonstrates the accuracy of the developed outage probability expression shown in Corollary \ref{corollary2}. 

A natural extension of Fig.  \ref{fig2} is to consider that the locations of the NOMA users follow a HPPP   in the whole cell, instead of on a line. However, for a cell with   radius being $1000$ m, the averaged number of the NOMA users is already above $3000$ even if $\lambda=0.01$, which makes it challenging to carry out  Monte Carlo simulations. Therefore, in Fig. \ref{fig3},  a Poisson cluster process is adopted \cite{Haenggi},  where the NOMA users are clustered within   a disc with radius $R_{c}$,  the center of the disc is uniformly located within the cell, and the   NOMA user with the best channel gain is scheduled. In addition, assume that   the legacy users are equally spaced and located on a semicircle with radius $50$ m. As can be seen from the figure, the scheduled NOMA user can be served with a reasonable outage probability, which means that the near-field beams preconfigured to   the legacy users can be used to efficiently serve NOMA users randomly located within the cell.  We note that increasing   $R$ can increase the achievable outage rate, but decrease the outage probability. Furthermore, it is important to point out that the use of more antennas can effectively increase the outage data rate of the NOMA user, an important property given the fact that massive MIMO is envisioned to be employed in future wireless networks.

  \begin{figure}[t] \vspace{-0.1em}
\begin{center}
\subfigure[ Outage Probabilites]{\label{fig3a}\includegraphics[width=0.33\textwidth]{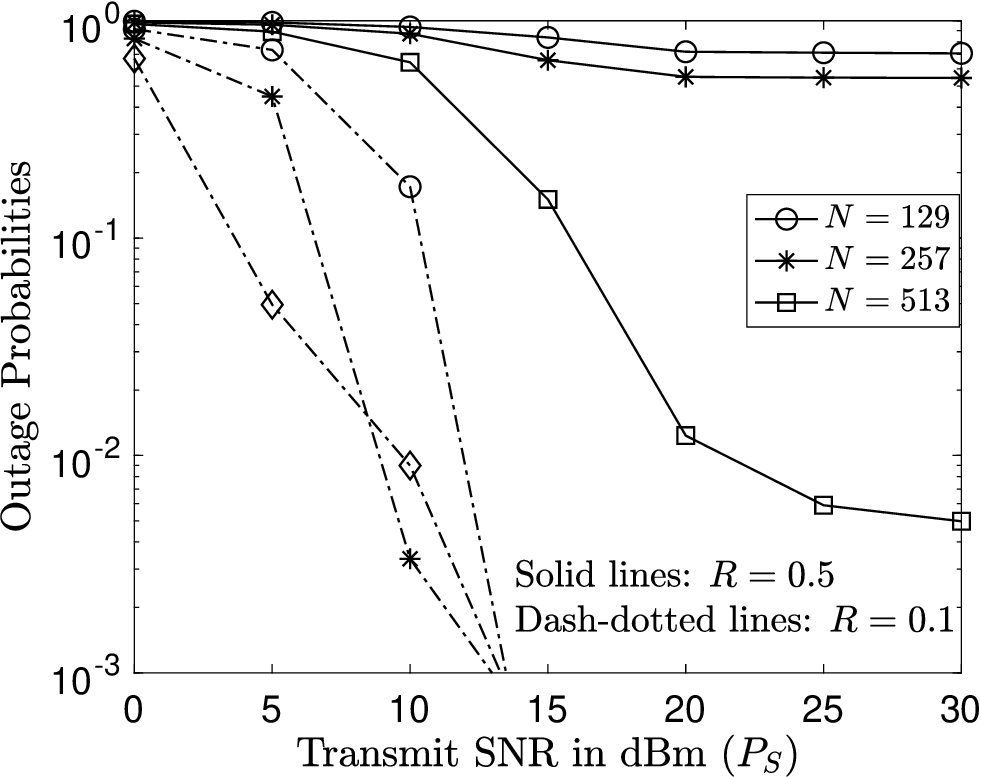}}\hspace{2em}
\subfigure[Outage Rates]{\label{fig3b}\includegraphics[width=0.33\textwidth]{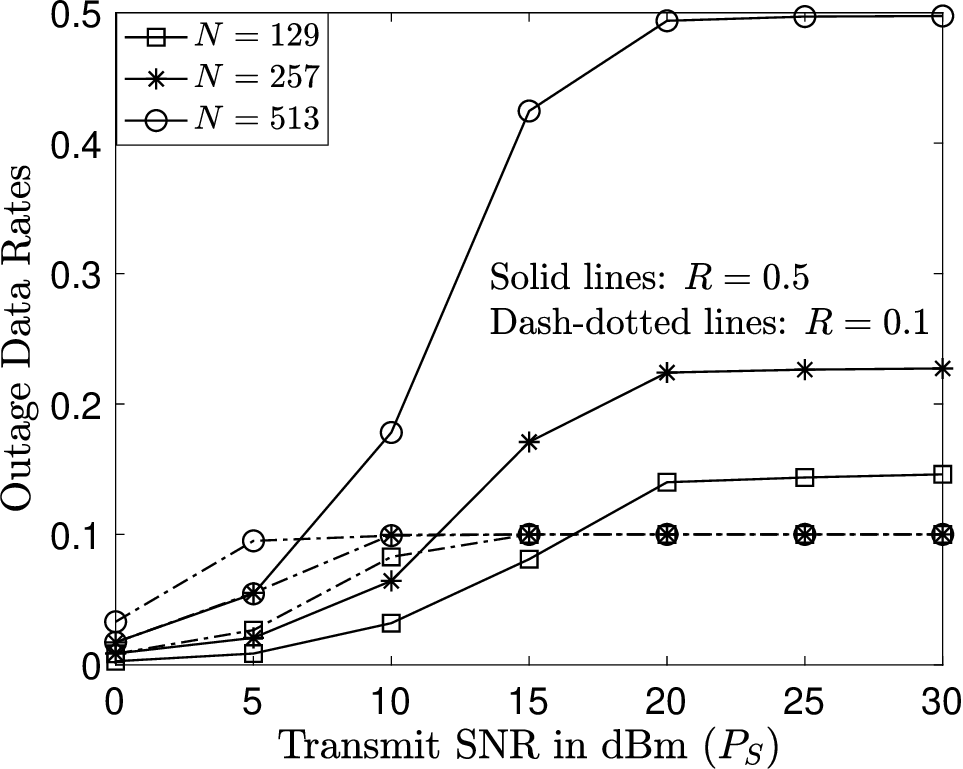}} \vspace{-1em}
\end{center}
\caption{Outage probability of the scheduled NOMA user by using the Poisson cluster process. $M=36$, $\lambda=0.05$, $r^{\rm L}=50$, $R_D=1000$, $R_c=10$,     $\alpha_m^{\rm N}=\frac{4}{5}$ and  $\alpha_m^{\rm L}=\frac{1}{5}$.  \vspace{-0.1em} }\label{fig3}\vspace{-2em}
\end{figure}

\vspace{-1em}
\section{Conclusions}
In this letter,   the condition under which the resolution of near-field beamforming is not perfect has been identified.   This imperfect resolution means that one user's near-field beam can be still useful to other users, which motivates the feasibility  study for the application of NOMA.   Both the   analytical and simulation results have been presented to demonstrate that those near-field beams preconfigured for legacy users can indeed be used  to effectively serve additional NOMA users, which improves the overall connectivity and system throughput.  We note that the analytical results developed in this letter are applicable for the special cases with $\tau\rightarrow 0$, where an important direction for future research is to develop a more general expression of $\Delta$. 
 \appendices
 \section{Proof for Lemma \ref{lemma1}}\label{proof1}
 By using the spherical-wave propagation model, the resolution of near-field beamforming, which is also the cross-correlation between the users' channel vectors, i.e., $\Delta \triangleq  \left|\mathbf{b}\left({\boldsymbol \psi}_1\right)^H\mathbf{b}\left({\boldsymbol \psi}_2\right)\right|^2$,   can be expressed as follows: 
\begin{align}
\Delta =  & \frac{1}{N^2}
\left|\sum^{N}_{n=1} e^{-j\frac{2\pi }{\lambda}\left(\left| {\boldsymbol \psi}_1 -{\boldsymbol \psi}_n\right|-\left| {\boldsymbol \psi}_2 -{\boldsymbol \psi}_n\right|\right)} \right|^2. 
\end{align}
Recall that each  distance can be expressed as $
\left| {\boldsymbol \psi}_m -{\boldsymbol \psi}_n\right| = r_m\sqrt{1+\frac{d_n^2-2r_md_n\sin\theta_m}{r_m^2}}$, for $m\in\{1,2\}$, 
where $d_n = d\left(n-1-\frac{N-1}{2}\right)$ and $d$ denotes the antenna spacing of the ULA. Note that $r_m$ is assumed to be proportional to $d_{\rm Ray}$, and hence $r_m\gg  d_n$, which means $
\frac{d_n^2-2r_md_n\sin\theta_m}{r_m^2}\rightarrow 0$ and hence leads to the following approximation \cite{BJORNSON20193, 9738442, ldmadai}: 
\begin{align}\label{appro1}
\left| {\boldsymbol \psi}_m -{\boldsymbol \psi}_n\right| 
 \approx  r_m\left(1-\frac{  d_n\sin\theta_m}{r_m}+\frac{ d_n^2(1-\sin^2\theta_m) }{2r_m^2}\right),
\end{align}
where  the following approximation, $(1+x)^{\frac{1}{2}}\approx 1+\frac{1}{2}x-\frac{1}{8}x^2$ for $x\rightarrow 0$, is used. 

By applying the approximation in \eqref{appro1}, $\Delta$ can be simplified as follows: 
\begin{align}\nonumber
\Delta\approx &  \frac{1}{N^2} \left|
\sum^{N}_{n=1} e^{-j\frac{2\pi }{\lambda}\left(\frac{ d_n^2(1-\sin^2\theta_0) }{2r_1}-\frac{ d_n^2(1-\sin^2\theta_0) }{2r_2}\right)} \right|^2
\\\label{approx2}
=&  \frac{1}{N^2}  \left|
\sum^{\frac{N-1}{2}}_{k=-\frac{N-1}{2}} e^{j\pi  k^2 d^2\frac{1}{\lambda} (1-\sin^2\theta_0) \left(\frac{  1}{r_1}-\frac{ 1}{r_2}\right)} 
\right|^2 .
\end{align}
In the literature, the above sum was shown to go to zero if $N\rightarrow \infty$, which is not accurate if the users' distances to the base station are proportional to the Rayleigh distance.

In particular, assume  that $r_1=\beta_1d_{\rm Ray}$ and $r_2=\beta_2d_{\rm Ray}$, where $d_{\rm Ray}$ denotes the Rayleigh distance and is given by $d_{\rm Ray}= \frac{2d^2(N-1)^2}{\lambda}$. By substituting the expressions of $r_1$ and $r_2$ to \eqref{approx2}, $\Delta$ can be expressed as follows:
\begin{align}
\Delta\approx&     \frac{1}{N^2} \left|
\sum^{\frac{N-1}{2}}_{k=-\frac{N-1}{2}} e^{j\pi  k^2 \frac{1}{  2(N-1)^2} (1-\sin^2\theta_0) \left(\frac{  1}{\beta_1}-\frac{ 1}{\beta_2}\right)} 
\right|^2 .
\end{align}

By using the expression of $\tau$, $\Delta$ can be expressed as follows: 
\begin{align}
\Delta  
=&   \frac{1}{N^2} \left|
\sum^{\frac{N-1}{2}}_{k=-\frac{N-1}{2}}  \cos\left(  \tilde{\theta}_k\right)+j\sin \left(\tilde{\theta}_k
\right)
\right|^2 ,
\end{align}
where $\tilde{\theta}_k=\frac{ \pi k^2}{  2(N-1)^2} \tau$.
If $\tau \rightarrow 0$, $\Delta$ can be further approximated as follows:  
\begin{align}
\Delta 
\overset{(a)}{\approx}&   \frac{1}{N^2} \left|N-
\sum^{\frac{N-1}{2}}_{k=-\frac{N-1}{2}}  \tilde{\theta}_k^2 +j \sum^{\frac{N-1}{2}}_{k=-\frac{N-1}{2}}  \tilde{\theta}_k
\right|^2 
\\\nonumber
\overset{(b)}{\approx}&      1-
\frac{2}{N}\sum^{\frac{N-1}{2}}_{k=-\frac{N-1}{2}}  \tilde{\theta}_k^2 +\frac{1}{N^2}\left(\sum^{\frac{N-1}{2}}_{k=-\frac{N-1}{2}} \tilde{\theta}_k\right)^2
\\\nonumber 
=&      1-
  \frac{ \pi^2\tau^2}{N  (N-1)^4}  \sum^{\frac{N-1}{2}}_{k=1}  k^4+   \frac{ \pi^2 \tau^2}{ N^2 (N-1)^4}  \left(\sum^{\frac{N-1}{2}}_{k=1}k^2\right)^2 ,
\end{align}
where step $(a)$      is obtained by using the   approximation: $\cos(x)\approx 1-x^2$ and $\sin(x)=x$ for $x\rightarrow 0$,  and step $(b)$ follows by applying the approximation $(1-x)^2\approx 1-2x$ for $x\rightarrow 0$. 

By using the two following finite sums,  $\sum_{k=1}^{n} k^4 = \frac{n(n+1)(2n+1)(3n^2+3n-1)}{30}$ and $\sum_{k=1}^{n} k^2 = \frac{n(n+1)(2n+1)}{6}$,   $\Delta$ can be further approximated as follows:
\begin{align}
\Delta 
\approx&    
       1-
  \frac{ \pi^2\tau^2}{ N (N-1)^4} \frac{N(N-1)(N+1)(3N^2-4)}{480}\\\nonumber &+   \frac{ \pi^2 \tau^2}{  N^2(N-1)^4}  \left(\frac{N(N-1)(N+1)}{24}\right)^2 .
\end{align}
With some straightforward algebraic manipulations, the approximated expression shown in the lemma can be obtained, and the proof is complete. 

\bibliographystyle{IEEEtran}
\bibliography{IEEEfull,trasfer}
  \end{document}